\documentclass[11pt,a4paper]{article}

\usepackage[utf8]{inputenc}

\usepackage{amsmath}
\usepackage{amsthm}
\usepackage{amssymb}
\usepackage{amsfonts}
\usepackage{nccmath}
\usepackage{bbm}
\usepackage[b]{esvect}
\usepackage{mathrsfs}
\usepackage{mathtools}
\usepackage{centernot}
\usepackage{xspace}
\usepackage{stmaryrd}

\usepackage{totcount}
\usepackage{color}
\usepackage[colorinlistoftodos]{todonotes}
\usepackage{verbatim}

\usepackage{graphicx}
\usepackage{lscape}
\usepackage{float}
\usepackage{supertabular}
\usepackage{multirow}
\usepackage{longtable}
\usepackage{enumerate}
\usepackage[inline]{enumitem}
\usepackage[framemethod=tikz]{mdframed} %and thus tikz
\usepackage{algpseudocode}

\usepackage{hyperref} %Load last if possible
\hypersetup{
   breaklinks=true,   % splits links across lines
}
\usepackage[hyphens]{xurl}
\usepackage[capitalize]{cleveref}

\usepackage{makecell}

\usepackage{thm-restate}

\usepackage{eulervm}  %<-------------- Add new packages here
\makeatletter
\DeclareRobustCommand*{\bfseries}{%
  \not@math@alphabet\bfseries\mathbf
  \fontseries\bfdefault\selectfont
  \boldmath
}
\makeatother

\makeatletter
\renewcommand{\paragraph}[1]{\medskip\noindent{\bf #1}}
\makeatother

\RequirePackage{totcount}
\RequirePackage{color}
\RequirePackage[colorinlistoftodos]{todonotes}

\newtotcounter{notecount}
\newcommand{\notewarning}{%
\ifnum\totvalue{notecount}>0%
 \vspace{1ex}
\begin{center}
 \begin{tikzpicture}[baseline=(A.south)]
    \node (A) [] at (0,0){};
    \node [rounded corners=1pt,rectangle, draw=red, fill=red!20,text=black](B) at (0.1ex,0ex){
        \Large \raggedright {\bf Warning:} There are still some notes left!
    };
 \end{tikzpicture}
\end{center}
 \vspace{1ex}
\fi
}
\makeatletter
\def\myaddcontentsline#1#2#3{%
  \addtocontents{#1}{\protect\contentsline{#2}{#3}{Section \thesubsection\ at p. \thepage}{}}}
\renewcommand{\@todonotes@addElementToListOfTodos}{%
    \if@todonotes@colorinlistoftodos%
        \myaddcontentsline{tdo}{todo}{{%
            \colorbox{\@todonotes@currentbackgroundcolor}%
                {\textcolor{\@todonotes@currentbackgroundcolor}{o}}%
            \ \@todonotes@caption}}%
    \else%
        \myaddcontentsline{tdo}{todo}{{\@todonotes@caption}}%
   \fi}%
\newcommand*\mylistoftodos{%
  \begingroup
       \setbox\@tempboxa\hbox{Section 9.9 at p. 99}%
       \renewcommand*\@tocrmarg{\the\wd\@tempboxa}%
       \renewcommand*\@pnumwidth{\the\wd\@tempboxa}%
       \listoftodos%
  \endgroup
}
\makeatother
\definecolor{lightgreen}{rgb}{0.86, 0.93, 0.78}
\definecolor{bordergreen}{rgb}{0.55, 0.76, 0.74}
\definecolor{lightblue}{rgb}{0.70, 0.90, 0.99}
\definecolor{borderblue}{rgb}{0.01, 0.66, 0.96}
\definecolor{lightamber}{rgb}{1, 0.93, 0.70}
\definecolor{borderamber}{rgb}{1, 0.76, 0.03}
\definecolor{lightcolor4}{rgb}{ 0.93, 0.70, 1}
\definecolor{bordercolor4}{rgb}{0.76, 0.03, 1}
\definecolor{lightcolor5}{rgb}{0.78,0.86,0.93}
\definecolor{bordercolor5}{rgb}{0.74,0.55,0.76}

\RequirePackage{algpseudocode}
\makeatletter
\algnewcommand{\ExtendedState}[1]{\State
\parbox[t]{\dimexpr\linewidth-\ALG@thistlm}{\hangindent=\algorithmicindent\strut\hangafter=3#1\strut}}
\makeatother
\algnewcommand\algorithmicinput{\textbf{Input:}}
\algnewcommand\Input{\item[\algorithmicinput]}

\algrenewcommand{\algorithmiccomment}[1]{{\color{gray}// #1}}
\algnewcommand{\IIf}[1]{\State\algorithmicif\ #1\ \algorithmicthen}
\algnewcommand{\EndIIf}{\unskip\ \algorithmicend\ \algorithmicif}

 \RequirePackage{varwidth}
 \RequirePackage{color}
 \RequirePackage[most]{tcolorbox}%with most option

 \newtcolorbox{titlebox}[5]{enhanced,center,colframe=black,colback=white,boxrule={#3},arc={#2},auto outer arc,%
 breakable,pad at break*=5pt,vfill before first,before={%\par\smallskip\noindent
 },after={\par\smallskip},top=12pt,left=4pt,%
 enlarge top by=2pt,%enlarge bottom by=7pt,%
 fontupper=\small,
 title={\rule[-.3\baselineskip]{0pt}{\baselineskip}\normalsize\sffamily\bfseries #1}, varwidth boxed title*=-30pt, 
 attach boxed title to top left={yshift=-10pt,xshift=10pt}, coltitle=black,
 boxed title style={colback=white,boxrule={#5},arc={#4},auto outer arc}
 }

\makeatletter
\let\orgdescriptionlabel\descriptionlabel
\renewcommand*{\descriptionlabel}[1]{%
  \let\orglabel\label
  \let\label\@gobble
  \phantomsection
  \edef\@currentlabel{#1}%
  \let\label\orglabel
  \orgdescriptionlabel{#1}%
}
\makeatother

\newcommand{\party}{P\xspace}
 
\newcommand{\unknowndelta}[0]{\textsc{UL}}
\newcommand{\gst}[0]{\textsc{GST}}

\newcommand{\clock}[0]{\textsc{Clock}}

\newcommand{\slowclock}[0]{\textsc{SlowClock}}

\newcommand{\realclock}[0]{\textsc{RealClock}}

\newcommand{\sentmessages}[0]{\textsc{SentMsgs}}
\newcommand{\receivedmessages}[0]{\textsc{ReceivedMsgs}}
\newcommand{\msg}{\textsc{msg}}

\newcommand{\events}[0]{\textsc{events}}
\newcommand{\protocols}[0]{\mathbold{\Pi}}

\usepackage[subtle]{savetrees}
\usepackage{cite}
\usepackage{lineno}

\newtheorem{theorem}{Theorem}
\newtheorem{lemma}[theorem]{Lemma}

\title{Unifying Partial Synchrony}
\author{Andrei Constantinescu \\ \href{mailto:aconstantine@ethz.ch}{aconstantine@ethz.ch} \and Diana Ghinea \\ \href{mailto:ghinead@ethz.ch}{ghinead@ethz.ch} \and Jakub Sliwinski \\ \href{mailto:jsliwinski@ethz.ch}{jsliwinski@ethz.ch} \and Roger Wattenhofer \\ \href{mailto:wattenhofer@ethz.ch}{wattenhofer@ethz.ch}}
\date{ETH Zürich, Switzerland}

\begin{document}

\maketitle

\begin{abstract}
	The distributed computing literature considers multiple options for modeling communication. Most simply, communication is categorized as either synchronous or asynchronous. Synchronous communication assumes that messages get delivered within a publicly known timeframe and that parties' clocks are synchronized. Asynchronous communication, on the other hand, only assumes that messages get delivered eventually.
	A more nuanced approach, or a middle ground between the two extremes, is given by the \emph{partially synchronous model}, which is arguably the most realistic option. This model comes in two commonly considered flavors:
	\begin{enumerate}[nosep,label=(\roman*)]
		\item The Global Stabilization Time ($\gst$) model: after an (unknown) amount of time, the network becomes synchronous. This captures scenarios where network issues are transient.
		\item The Unknown Latency ($\unknowndelta$) model: the network is, in fact, synchronous, but the message delay bound is unknown.
	\end{enumerate}
    This work formally establishes that any time-agnostic property that can be achieved by a protocol in the $\unknowndelta$ model can also be achieved by a (possibly different) protocol in the $\gst$ model. By time-agnostic, we mean properties that can depend on the order in which events happen but not on time as measured by the parties. Most properties considered in distributed computing are time-agnostic. The converse was already known, even without the time-agnostic requirement, so our result shows that the two network conditions are, under one sensible assumption, equally demanding.
	% This work formally establishes that the $\unknowndelta$ model subsumes the $\gst$ model, demonstrating their equivalence.
\end{abstract}

% Gemini version 
% Within the distributed computing domain, various communication modeling paradigms exist. Traditionally, communication is categorized as either synchronous or asynchronous. Synchronous communication assumes synchronized clocks among participating entities and message delivery within a publicly known timeframe. Asynchronous communication, on the other hand, solely guarantees eventual message delivery.

% A more nuanced approach emerges with the partially synchronous model, arguably the most realistic representation. This model encompasses two primary variants:

% (i) Global Stabilization Time (GST) Model: This model posits an unknown point in time after which the network transitions to a synchronous state. This captures scenarios where network issues are transient.

% (ii) Unknown Δ Model: Here, communication delays are bounded by a constant Δ, though the specific value of Δ remains unknown. This variant reflects scenarios with inherent uncertainty in message transmission times.

% This work formally establishes that the Unknown Δ model subsumes the GST model, demonstrating their equivalence. In essence, we prove that accounting for inherent uncertainty in communication delays offers equivalent expressive power to the GST model's eventual transition to synchrony.

\section{Introduction}

Distributed computing systems underpin a vast array of contemporary technological advancements, ranging from cloud computing platforms to blockchain networks.  These systems rely on protocols to ensure consistency and reliability even when faced with challenges such as message delays and node failures. A cornerstone of designing robust protocols lies in understanding the communication model assumed by the distributed system. Within the distributed computing literature, the synchronous and asynchronous communication models remain the two best-established paradigms.
%%% OLD
% Distributed systems are the backbone of many modern technologies, from cloud computing to blockchain networks. These systems rely on protocols to ensure consistency and reliability even when faced with challenges such as message delays and node failures. Designing robust protocols requires understanding the underlying communication model. The most common models in the literature are the synchronous model and the asynchronous model. 
%%% GEMINI
% Distributed computing systems underpin a vast array of contemporary technological advancements, ranging from ubiquitous cloud computing platforms to burgeoning blockchain networks. These systems orchestrate the collaborative efforts of geographically dispersed processing units, necessitating the adoption of well-defined communication protocols to guarantee data consistency and system reliability in the face of potential disruptions, encompassing message delays and node failures. The cornerstone of designing robust protocols lies in a thorough comprehension of the underlying communication model employed by the distributed system. Within the distributed computing literature, the synchronous and asynchronous communication models remain the most extensively explored paradigms.
The synchronous model assumes a publicly known upper bound $\Delta$ on message delays and that parties hold synchronized clocks. 
This idealized setting facilitates the design of elegant protocols that can leverage a round-based structure and often achieve very high resilience thresholds. However, the synchronous model exhibits a fundamental limitation: any deviation from the assumed message delay bound $\Delta$  can render synchronous protocols entirely ineffective, potentially leading to complete breakdowns of the protocols.

The asynchronous model, on the other hand, makes no assumptions except that messages get delivered eventually. 
This inherent flexibility empowers the asynchronous model to support 
%highly robust
protocols that can gracefully adapt to any network conditions. However, asynchronous protocols typically exhibit lower resilience thresholds compared to their synchronous counterparts. Furthermore, even achieving agreement when parties might crash is impossible without randomization \cite{FLP}.
Hence, neither of these two extremes perfectly captures real-world systems: the synchronous model's assumptions are too strong, while the asynchronous model is too pessimistic.  
In this work, we are concerned with a middle ground between the two --- the \emph{partially synchronous} model, a nuanced paradigm that bridges the gap between the two, introduced by Dwork, Lynch, and Stockmeyer \cite{JACM:DLS88}. The work of \cite{JACM:DLS88} proposes two definitions for the partially synchronous model, described below (see the next section for the fully formal definitions).

\paragraph{The Global Stabilization Time ($\gst$) model.} This variant acknowledges that there might be periods of unpredictable delays due to network congestion or outages but also assumes that these disruptions eventually resolve and the system stabilizes. \cite{JACM:DLS88} explains how this intuition can be faithfully captured with a simple model, known as the Global Stabilization Time ($\gst$) model: there is an unknown `Global Stabilization Time' $T$ after which the system behaves synchronously for a publicly-known message delivery bound $\Delta$. In particular, there is a publicly known amount of time $\Delta$ such that every message sent at time $t$ is delivered by time $\max(t, T) + \Delta$; i.e., messages sent at time $t \geq T$ are delivered by time $t + \Delta$, while messages sent at time $t < T$ are delivered by time $T + \Delta$.

\paragraph{The Unknown Latency ($\unknowndelta$) model.} 
In this variant, the system, is in fact, always synchronous: there is a value $\Delta$ such that every message sent by time $t$ is delivered by time $t + \Delta$. However, as opposed to the synchronous model, the value of $\Delta$ is unknown to the protocol.

\paragraph{The relationship between the two.} The two models are conjectured to be equivalent, in the sense that any property that can be achieved by a protocol in one can also be achieved in the other. In fact, there is an elegant folklore reduction from the $\unknowndelta$ model to the $\gst$ model, presented in \cite{ittaiBlog}, which we explain next.
% \asnote{We need to find who actually put this in a paper first.}
Consider a protocol $\Pi$ achieving some property $X$ in the $\unknowndelta$ model. Let us run $\Pi$ in the $\gst$ model, where a value of $\Delta$ is provided and guaranteed to hold eventually, but $\Pi$ ignores it. Consider any execution $\varepsilon$ of $\Pi$ in this setting: by the guarantees of the model, there exists a time $T$ such that all messages in $\varepsilon$  sent at time $t$ get delivered by time $\max(t, T) + \Delta$. Hence, in $\varepsilon$ all messages get delivered within time $T + \Delta$, so from the perspective of the parties, they might just as well be running in $\unknowndelta$ with the unknown bound on message delay being $T + \Delta$. Hence, $\varepsilon$ is also a legal execution of $\Pi$ in $\unknowndelta$. As a result, the set of executions of $\Pi$ in $\gst$ is a subset of its set of executions in $\unknowndelta$. Since $\Pi$ satisfies $X$ in $\unknowndelta$, it also satisfies $X$ in $\gst$. 
% Roughly, one may run a protocol designed for the $\unknowndelta$ model in the $\gst$ model instead, where the assumed message delay is $\Delta$ and the unknown $\gst$ is time $T$. This is equivalent to an execution in the $\unknowndelta$ model where the unknown message delay is $\max(\Delta, T)$.
The same blog post \cite{ittaiBlog} also explains the reverse direction: a protocol designed for the $\gst$ model can be transformed into an equivalent protocol for the $\unknowndelta$ model, but only if it satisfies a certain property, namely, that the protocol's guarantees are still maintained if the assumed value of $\Delta$ changes dynamically. This way, one may increment the assumed $\Delta$ whenever a timeout of the protocol expires, and eventually, the assumed $\Delta$ will exceed the real one. However, as \cite{ittaiBlog} points out, assuming this property is not without loss of generality. To the best of our knowledge, whether the converse holds is still an open question. 

\paragraph{Our contribution.} In this work, we answer this question in the affirmative under the relatively minor technical assumption of only considering `time-agnostic' properties. A protocol property is \emph{time-agnostic} if whether it holds for a given execution of a protocol can only depend on the relative order in which events happened, but not on time as measured by the parties in the execution. We note that most properties considered in distributed computing are indeed time-agnostic; e.g., whether some consensus protocol satisfies given agreement, validity and termination conditions. Bounds on message complexity can also be accommodated, but the same is not true about running time guarantees. Additionally, we will only show our result assuming that the environment provides a global perfect clock to the parties, that is their only way of telling time. Our proof seems straightforward, albeit technical, to adapt for imperfect party clocks. On a similar note, we consider randomized protocols, but do not consider probabilistic properties, such as ``with probability at least $0.5$ all parties terminate''.\footnote{E.g., a randomized byzantine agreement protocol is required to always satisfy the agreement property.} We leave a formal argument considering imperfect clocks and probabilistic properties for future work. On the other hand, our proof works in adversarial settings, like when nodes may crash or deviate arbitrarily from the protocol. The key idea in our proof is that, instead of estimating the actual value of $\Delta$ in the $\unknowndelta$ model, like in the proof idea mentioned in the blog, we continuously slow down the parties' clocks. This is achieved by the parties applying a wrapper function on top of the time measurements returned by the system clock. This way, the parties simulate running the $\gst$ protocol with a continuously increasing value of $\Delta,$ which will eventually exceed the actual unknown $\Delta$ holding in the $\unknowndelta$ model, hence allowing the guarantees of the $\gst$ protocol that we are running to apply.

% in \emph{system time} becomes, in fact, a larger value if measured in \emph{real time}, and it eventually becomes larger than the actual message delay, allowing the guarantees of the $\gst$ protocol that we are running to apply.
\section{Preliminaries}

%Consider a protocol $\Pi$ designed in the GST model. Hence, parties are aware of a bound $\Delta_{\assumed}$ on the message delivery time, and, intuitively, they expect this bound to hold eventually, or, more concretely, after the $\gst$  passes.
%We will run the protocol $\Pi$ in the unknown-$\Delta$ model. To ensure that its guarantees are maintained, we continuously \emph{slow down} the parties' clocks. This way, eventually, the parties' clocks slow down enough so that the measured $\Delta_{\assumed}$ is at most the real unknown bound $\Delta_{\real}$. The moment this happens will be the Global Stabilization Time that $\Pi$ expects to hold eventually.
%In the remainder of this section, we formalize this idea.

We consider a fixed set of $n$ parties in a network, where links model communication channels. 
The parties are running a (possibly randomized) protocol over the network. For each party, the protocol is specified by a state transition diagram, where a party's state is defined by its local variables. The initial state of a party is then defined by any initial inputs and random coins. The transitions are deterministic (but may depend on the party's random coins). Without loss of generality, a party's transition to another state is triggered by the receipt of a message, or specific changes in time (e.g., waiting a predefined amount of time). State changes are instantaneous and include all required local computations and the corresponding sending of messages (i.e., these instructions are executed atomically). The receipt of messages, on the other hand, will be controlled by the \emph{message system}, which we discuss below.

\paragraph{Messages.} 
Messages are held in a global \emph{message system}: this maintains a set containing tuples $(P_s, P_r, m, c)$, where $P_s$ is the sender of the message, $P_r$ is a receiver of the message, $m$ is the content of the message, and $c$ is a unique identifier assigned by the message system.
% We add that these tuples may also model interactions with the environment, such as requests or transactions coming from clients or responses sent by the parties to the clients.
The message system is controlled by the adversary and may decide when to deliver these messages (subject to the constraints of the communication model). For simplicity of presentation, we assume that the message system keeps delivering messages even after the receivers have terminated (if the protocol allows it, or, e.g., if we consider crash faults). Otherwise, claims of the form ``eventually all messages get delivered within $\Delta$ time units'' would not be meaningful for terminated receivers.

% We also note that the protocol may allow parties to irreversibly decide on some output (i.e., if the protocol guarantees Termination as opposed to Liveness).

\paragraph{Global clock.} We assume that parties have access to a common \emph{global clock} denoted by $\clock$, which represents their only source of time. Abstractly, $\clock$ is represented by an increasing and continuous function $\clock : \mathbb{R}_{\geq 0} \to \mathbb{R}_{\geq 0}$ that maps real time to \emph{system time}. In particular, at real time $t,$ the parties can atomically query the global clock to read off a `system time' of  $\clock(t).$ Neither the parties nor the adversary have access to the actual definition of the function $\clock$. Instead, they can only use the global clock as an oracle to receive the current \emph{system time}. Depending on the environment, the system time may coincide with real time, in which case $\clock$ is the identity function $\realclock(t) = t$, but this is not necessarily the case.

% For simplicity of presentation, we assume that the system time coincides, in fact, with real time. However, as our proofs manipulate the responses of the global clock, our write-up carefully distinguishes between time measured in \emph{system time} and \emph{real time}.

% Hence, whenever a transition depends on time changes, this depends on the time returned by the shared global clock. For simplicity, we will assume that the time returned by the global clock coincides with real time: when receiving a query at real time $t$, it returns the system time $t$.

% Formally, the global clock is an oracle equipped with a continuous increasing function $\clock : \mathbb{R}_{\geq 0} \to \mathbb{R}_{\geq 0}$ that maps real time to \emph{the system's time}. The parties do not have access to the actual definition of the function $\clock$. Instead, the global clock allows them to atomically query the current system time. If the global clock receives a query at \emph{real time} $t$, it responds with the system time $\clock(t)$. The system time can coincide with real time, and in that case $\clock$ is the identity function $\realclock(t) = t$. 

\paragraph{Protocol execution models.} A protocol execution model $M$ captures all requirements and guarantees of the environment under which a protocol runs. For the purposes of our paper, communication in $M$ always happens through message passing as already described, and, moreover, $M$ specifies a global clock function $\clock$ that the parties use to tell time. Other aspects of the execution environment can appear as part of the guarantees of $M$, such as bounds on the message delay or other timing constraints. Two examples of such models $M$ are the $\gst(\Delta, \clock)$ and $\unknowndelta(\Delta, \clock)$ models, to be formally introduced below. Note that the guarantees of a fixed model $M$ are concrete: e.g., messages are delivered within $\Delta$ time units for a fixed $\Delta$; in contrast, often in the literature, models usually refer to families of models (in this particular case, parameterized by $\Delta$). Last but not least, $M$ specifies the power of the adversary. Other than controlling the scheduler within the timing constraints of the model, the adversary might additionally, for instance, make parties crash, fail to send certain messages or deviate from the protocol arbitrarily (i.e., byzantine behavior). Model $M$ should specify precisely which faults are possible and under what circumstances (e.g., whether the adversary is adaptive, computationally bounded, and how many parties it can corrupt). Importantly, the parties are not aware of the clock function used: from their perspective, this is supplied by the environment as an oracle, with no access to its implementation. More abstractly, a model $M$ specifies for each protocol $\Pi$ its set of legal executions $\varepsilon$, defined in the following.

\paragraph{Executions.} Consider a protocol $\Pi$ running in a model $M$ where parties measure time using function $\clock$. An execution of $\Pi$ is defined by the parties' initial states and a (possibly infinite) collection of events,
denoted by $\events(\varepsilon)$. 
Each event in $\events(\varepsilon)$ is a tuple $(t,$ $\receivedmessages,$ $\party,$ $q,$ $\sentmessages)$ signifying that, at system time $t$ (i.e., as observed by the parties using function $\clock$), party $\party$ received the (possibly empty) multiset of messages $\receivedmessages$ from the message system, $\party$'s state became $q$ (possibly the same state it was already in), and $\party$ sent the (possibly empty) multiset of messages $\sentmessages$ to the message system. We say that a message $\msg = (P_s, P_r, m, c)$ was sent at system time $t$ in execution $\varepsilon$ if $\events(\varepsilon)$ contains some event $(t, \receivedmessages, \party_s, q, \sentmessages)$ with $\msg \in \sentmessages$.
Similarly, we say that a message $\msg = (P_s, P_r, m, c)$ was received at system time $t$ in execution $\varepsilon$ if $\events(\varepsilon)$ contains some event $(t, \receivedmessages, \party_r, q, \sentmessages)$ with $\msg \in \receivedmessages$. Note that a message sent/received at system time $t$ is sent/received at real time $\clock^{-1}(t)$. We have made the deliberate choice to timestamp executions in system time as this is the perspective that parties perceive them from. This will allow us to map between executions with different clock functions in our main result.

\paragraph{The $\gst$ model.} The $\gst$ model has as parameters a clock function $\clock$ that the environment provides to the parties to tell the time when running a protocol, and $\Delta$, to be supplied to protocols designed for the model when instantiated for a specific $\Delta$. We write $\gst(\Delta, \clock)$ for the model instantiated with specific parameters $\Delta$ and $\clock$. The model guarantees that, for every protocol $\Pi$ and every execution $\varepsilon$ of $\Pi$ in the model, there exists a time $T$ measured in real time such that every message in $\varepsilon$ sent at real time $t$ is received by real time $\max(t, T) + \Delta$. The model can be altered to give the adversary more power than controlling the scheduler; e.g., to corrupt parties.

\paragraph{The $\unknowndelta$ model.} The $\unknowndelta$ model has as parameters a clock function $\clock$ that the environment similarly provides to the parties, and $\Delta$, not to be supplied to protocols designed for the model. We write $\unknowndelta(\Delta, \clock)$ for the model instantiated with specific parameters $\Delta$ and $\clock$. The model guarantees that, for every protocol $\Pi$ and every execution $\varepsilon$ of $\Pi$ in the model, any message in $\varepsilon$ sent at real time $t$ is received by real time $t + \Delta$. 
%{\color{red}{Easier version: The model guarantees that every message sent at real time $t$ is received by real time $t + \Delta$.}} 
This model can also be altered to give more power to the adversary.

\paragraph{Protocol properties.} We define a protocol property as a set of allowed executions; e.g., the property that all parties eventually terminate, or that they produce some outputs. We say that a protocol achieves a property in a model $M$ if all its legal executions in that model satisfy the property, i.e., are in the set of executions allowed by the property. 
% Note that modeling certain properties this way is non-trivial, as executions alone do not contain, e.g., who are the byzantine parties and when they were corrupted. However, even such properties can be modeled: parties that are byzantine and behave honestly are indistinguishable from honest parties, so all information required to define such properties is found in the execution. Alternatively, one can modify executions to include corruption events to make the process more transparent.
% \textcolor{magenta}{
Note that modeling certain properties this way is non-trivial, as executions alone do not contain, e.g., who are the byzantine parties and when they were corrupted.
However, even such properties can be modeled: executions may contain changes of states that do not follow from the protocol’s state transition to model parties misbehaving, or one can modify executions to include corruption events to make the process more transparent. In this paper, we are concerned with time-agnostic properties, defined next. We call two executions $\varepsilon, \varepsilon'$ \emph{equivalent} if they differ only in the timestamps of the events and agree on the relative order of the events. A property $X$ is \emph{time-agnostic} if for any two equivalent executions $\varepsilon, \varepsilon'$ it holds that $\varepsilon \in X \iff \varepsilon' \in X$.
% }

% {\color{blue} [Maybe include some of this old content here if it's not convincing enough:] 
% Although corrupted parties may deviate from the protocol, their actions are also captured by an execution's events. 
% We say that a party $\party$ is faulty in execution $\varepsilon$ of a protocol if $\events(\varepsilon)$ contains an event $(t, \receivedmessages, \party_s, q, \sentmessages)$ that does not follow from the protocol's transition diagram.
% For example, if a corrupted party $\party$ crashes at system time $t$, any events $(t', \receivedmessages, \party_s, q, \sentmessages)$ with $t' \geq t$ will satisfy $\sentmessages = \emptyset$.}

% \paragraph{Properties of $\gst$ protocols.} Note that protocols designed for the $\gst$ model are colloquially still called protocols, while in reality being protocol families, one for each potential value of $\Delta$. Let $\protocols$ be the set of all protocols. A protocol family $\Pi : \mathbb{R}_{\geq 0} \to \protocols$ achieves a property in $\gst(\Delta, \clock)$ iff all executions of $\Pi(\Delta)$ achieve this property in $\gst(\Delta, \clock)$. One might wonder whether this is a meaningful definition when $\clock$ is not a perfect clock; e.g., if $\clock$ is twice slower than it needs to be, it seems natural to instead run the protocol as $\Pi(\Delta / 2)$. However, it is the responsibility of the protocol designer to adapt the protocol for the expected clock type.

\paragraph{Augmented models.} 
Our result will be very general: we will consider an arbitrary protocol $\Pi$ designed for the $\gst$ model, instantiated with a publicly-known eventual message delay bound of 1, that satisfies a given time-agnostic property in $\gst(1, \realclock)$. We will show how $\Pi$ can be transformed into a protocol $\Pi'$ that only depends on $\Pi$ that satisfies the same property in $\unknowndelta(\Delta, \realclock)$, irrespective of the value of $\Delta$. Moreover, if we augment both the $\gst$ model and the $\unknowndelta$ model with the same kind of additional power for the adversary, the same statement holds, with the same proof. For simplicity, in the following, we assume the basic models, but we note that we also get the result for a plethora of more interesting fault settings, e.g., byzantine faults and crashes.

\section{Our Reduction} 

This section presents the proof of our main result, stated below.

\begin{theorem} \label{thm:main} Any time-agnostic property that can be achieved by a protocol in the $\gst$ model can also be achieved by a protocol in the $\unknowndelta$ model.
\end{theorem}

% As already suggested, the key idea behind our reduction will be to slow down the parties' clocks. In the following, we first formally describe the model. Afterwards, we ensure that slowing down the clocks does not affect the protocol's guarantees. Finally, we show that using slow clocks allows us to run a protocol designed for the $\gst$ model in the $\unknowndelta$ model while maintaining the protocol's guarantees.

As previously mentioned, the key idea behind our reduction will be slowing down time.
%, or more precisely, the system time.
Given a protocol $\Pi$ achieving some time-agnostic property $X$ in $\gst(1, \realclock),$ we construct a protocol $\Pi'$ such that any execution $\varepsilon'$ of $\Pi'$ in $\unknowndelta(\Delta, \realclock)$ for some $\Delta$ unbeknownst to the protocol is equivalent to a legal execution $\varepsilon$ of $\Pi$ in $\gst(1, \realclock),$ hence also achieving property $X$.

Protocol $\Pi'$ will simulate running protocol $\Pi$ with a modified system clock that continuously slows down, so that equal intervals of time measured in the simulated system will represent longer and longer spans of real time. Moreover, we need that the modified system clock eventually gets arbitrarily slow. This way, since protocol $\Pi$ is designed to have property $X$ under the assumption that, once a sufficient amount of time passes, every message gets delivered within $1$ unit of system time, this will eventually be the case: the clock will get slow enough for $1$ unit of system time to correspond to a span of real time exceeding the unknown message delay bound.

More specifically, $\Pi'$ will simulate $\Pi$ running with system clock $\slowclock  : \mathbb{R}_{\geq 0} \to \mathbb{R}_{\geq 0}$ given by $\slowclock(t) = \sqrt{t}$. To achieve the simulation, whenever a party in the simulated $\Pi$ queries the global clock and the answer would have normally been $t$ (coinciding with the real time), $\Pi'$ replaces the answer with $\slowclock(t)$. This way, from the perspective of the simulated $\Pi$, the system clock is $\slowclock$. We note that it would have sufficed to take $\slowclock$ to be any increasing function whose derivative tends to 0 as $t \to \infty$ to achieve the required apparent slowdown of time.

% We make use of $\slowclock(t) = \sqrt{t}$, but any sublinear bijection suffices for our purpose. \asnote{Me and Diana had different notions of sublinear, check.}

The first lemma below shows the required result assuming that $\Pi$ is running standalone but with system clock $\slowclock$. The second lemma lifts it to the protocol $\Pi'$ that runs with system clock $\realclock$, but simulates $\Pi$ running with system clock $\slowclock$. A short discussion of why this implies
\cref{thm:main} follows.

\begin{lemma} \label{lemma:main-lemma}
    Consider a protocol $\Pi$ and an execution $\varepsilon$ of $\Pi$ in $\unknowndelta(\Delta,$ $\slowclock)$. Then, $\varepsilon$ is also an execution of $\Pi$ in $\gst(1, \realclock)$.
\end{lemma}
\begin{proof}
Consider an execution $\varepsilon$ of $\Pi$ in $\unknowndelta(\Delta, \slowclock)$, which guarantees that any message sent at real time $t$ is delivered by real time $t + \Delta$. From the perspective of the parties, however, time is measured using $\slowclock,$ so in $\varepsilon$, the parties observe that any message sent at system time $\slowclock(t)$ is delivered by system time $\slowclock(t + \Delta).$ Let us consider $\slowclock(t + \Delta) - \slowclock(t) = \sqrt{t + \Delta} - \sqrt{t}$ to understand how the message delay observed by the parties evolves with $t$. Taking the derivative, the function is strictly decreasing with $t,$ so the observed network delay gets smaller and smaller as time passes. Subsequently, let us find a bound $t_0$ on $t$ such that starting at real time $t_0$, the observed network delay is bounded by 1; i.e., let us solve $\sqrt{t + \Delta} - \sqrt{t} \leq 1$. If $\Delta < 1,$ this happens for $t \geq 0$. Otherwise, $\Delta \geq 1,$ and this happens for $t \geq \frac{1}{4}(\Delta - 1)^2$. Hence, starting at real time $\frac{1}{4}(\max\{1, \Delta\} - 1)^2$, the observed (system) network delay is bounded by 1. Writing the same in terms of system time, starting at system time $T := \sqrt{\frac{1}{4}(\max\{1, \Delta\} - 1)^2} = \frac{1}{2}(\max\{1, \Delta\} - 1)$, the system network delay is bounded by 1. In particular, this means that a message sent at system time $t$ in $\varepsilon$ is delivered by system time $\max\{t, T\} + 1$. Hence, $\varepsilon$ could just as well be an execution of $\Pi$ in $\gst(1, \realclock)$ with global stabilization time $T$ because the parties and the adversary are unaware of the clock function used.
\end{proof}

\begin{lemma} \label{lemma:almost-there} Consider a protocol $\Pi$ achieving some time-agnostic property $X$ in $GST(1,$ $\realclock)$. Then, there is a protocol $\Pi'$ depending only on $\Pi$ that archives $X$ in $\unknowndelta(\Delta,$ $\realclock)$ for all $\Delta \geq 0$. 
\end{lemma}
\begin{proof} In protocol $\Pi'$ parties run protocol $\Pi$ but apply the function $\slowclock : \mathbb{R}_{\geq 0} \to \mathbb{R}_{\geq 0}$ as a wrapper over the global clock's responses to the queries. In particular, whenever a party queries the global clock in $\Pi$ and the time returned is $t$, the party evaluates $\slowclock(t)$ and takes this as the answer instead. Every execution $\varepsilon'$ of $\Pi'$ in some model $M$ corresponds to an equivalent execution $\varepsilon$ of $\Pi$ in $M$ where the clock function provided by the environment is composed with $\slowclock$. Namely, every event $e$ present in $\varepsilon$ and $\varepsilon'$ is timestamped $t$ in $\varepsilon$ and $\slowclock(t)$ in $\varepsilon'$.

Hence, for any $\Delta \geq 0,$ any legal execution $\varepsilon'$ of $\Pi'$ in $\unknowndelta(\Delta, \realclock)$ corresponds to a legal execution $\varepsilon$ of $\Pi$ in $\unknowndelta(\Delta, \slowclock)$ that is equivalent to $\varepsilon'$. By \cref{lemma:main-lemma}, $\varepsilon$ is also a legal execution of $\Pi$ in $\gst(1, \realclock)$. Since $\Pi$ achieves property $X$ in $\gst(1, \realclock),$ it follows that $\varepsilon \in X$, and hence, since $X$ is time-agnostic, $\varepsilon' \in X$. Since $\varepsilon'$ was an arbitrary execution of $\Pi'$ in  $\unknowndelta(\Delta, \realclock)$ and $\Delta \geq 0$ was arbitrary, it follows that $\Pi'$ satisfies property $X$ in $\unknowndelta(\Delta, \realclock)$ for all $\Delta \geq 0$.
 \end{proof}

% We are now ready to discuss how the previous lemma readily implies our main theorem.

\begin{proof}[Proof of Theorem \ref{thm:main}]

Formally, if we write $\protocols$ for the set of all protocols, a protocol designed for the $\gst$ model is, in fact, a protocol family $\Pi : \mathbb{R}_{\geq 0} \to \protocols$, one for each potential value of the publicly-known eventual message delay bound. $\Pi$ achieves a property in $\gst(\Delta, \realclock)$ for some $\Delta \geq 0$ iff all executions of $\Pi(\Delta)$ achieve this property in $\gst(\Delta, \realclock)$. 

Let $\Pi$ be a protocol achieving some time-agnostic property $X$ in the $\gst$ model. For our proof, we only need the fact that $\Pi(1)$ satisfies $X$ in $\gst(1,$ $\realclock)$. Applying \cref{lemma:almost-there} to $\Pi(1)$, we get that $\Pi'$ satisfies $X$ in $\unknowndelta(\Delta, \realclock)$ for all $\Delta \geq 0,$ implying the conclusion.
\end{proof}

\bibliographystyle{plainurl}
\bibliography{project}

\end{document}